\documentclass[11pt]{article}

\usepackage{authblk}
\usepackage[margin=1in]{geometry}
\usepackage{amsmath,amssymb,amsthm,mathtools}
\usepackage{booktabs}
\usepackage{graphicx}
\usepackage{microtype}
\usepackage{natbib}
\usepackage{enumitem}
\usepackage{url}
\usepackage{placeins}
\usepackage{algorithm}
\usepackage[hidelinks]{hyperref}

\setlist{nosep}
\graphicspath{{figures/}{./}}

\newtheorem{theorem}{Theorem}
\newtheorem{proposition}[theorem]{Proposition}
\newtheorem{corollary}[theorem]{Corollary}
\newtheorem{lemma}[theorem]{Lemma}
\theoremstyle{remark}
\newtheorem{remark}{Remark}

\newcommand{\E}{\mathbb{E}}
\newcommand{\Var}{\operatorname{Var}}
\newcommand{\Prb}{\mathbb{P}}
\newcommand{\ind}{\mathbf{1}}
\newcommand{\given}{\,\middle|\,}

\title{Adaptive Sampling of Costly Outcomes in Randomized Clinical Trials}

\author[1,*]{Shuoyang Wang}
\author[2,3,*]{Wanyu Zhang}
\author[4]{Lingli Yang}
\author[5,\textdagger]{Kexuan Li}

\affil[1]{Department of Bioinformatics and Biostatistics, University of Louisville, Louisville, KY, USA}
\affil[2]{Division of Gastroenterology, Hepatology and Nutrition, University of Louisville School of Medicine, Louisville, KY, USA}
\affil[3]{Louisville Clinical and Translational Research Center, University of Louisville, Louisville, KY, USA}
\affil[4]{Takeda Pharmaceuticals, Cambridge, MA, USA}
\affil[5]{Data \& Quantitative Sciences, Bristol Myers Squibb, Cambridge, MA 02141, USA}

\date{}

\begin{document}
\maketitle

\renewcommand{\thefootnote}{\fnsymbol{footnote}}
\footnotetext[1]{These authors contributed equally to this work as the co-first authors.}
\footnotetext[2]{Corresponding author: Kexuan Li (\texttt{kexuan.li.77@gmail.com}).}
\renewcommand{\thefootnote}{\arabic{footnote}}

\begin{abstract}
In some randomized trials the primary outcome is costly or slow to measure, while auxiliary variables that predict it are available for everyone. The outcome can then be measured in a probability sample. We study an adaptive design in which the statistician who selects outcomes to measure is blinded to treatment assignment. Outcomes from an initial random sample are used to fit pooled models for the outcome and its residual variance. These set the sampling probabilities for the remaining participants and may be refitted as outcomes accumulate. After unblinding, arm means are estimated by augmented inverse probability weighting. The estimator is unbiased for the complete-data treatment difference for any working models, and a martingale central limit theorem gives Wald intervals under repeated updating. We bound the variance lost by estimating the sampling rule. The bound is linear in the error of the fitted residual variance, quadratic when the optimal probabilities are not truncated, and relates the initial sample size to the learning rate of the models. Relative to designs using treatment assignment, the blinded design loses a term due to unequal residual variances in the two arms and a term due to the conditional treatment effect, which is second order near the null. In simulations, coverage was near nominal. Adaptive sampling was 12\% to 34\% more efficient than simple random sampling and needed 10\% to 26\% fewer measured outcomes for the same precision. In a resampling study of an antifungal trial, adaptive sampling reduced the sampling variance by 39\%.
\end{abstract}


\section{Introduction}

Some outcomes in clinical trials take much more time or money to obtain than the data collected during routine follow-up. Examples include specialized laboratory assays, central review of imaging, adjudication of clinical events by a committee, review of medical records, and detailed clinical assessments. Central adjudication is standard in many large trials, although its value relative to its cost has been questioned \citep{GrangerEtAl2008}. When variables that predict the costly outcome are available for all randomized participants, the outcome can instead be measured in a probability sample. This reduces the number of measurements and keeps a valid comparison of the randomized arms. In oncology, for example, \citet{DoddEtAl2011} proposed central review of progression for a random sample of patients, combined with the local assessments of all patients through an auxiliary-variable estimator.

The problem is a form of two-phase sampling, which has a long history in survey sampling and semiparametric inference \citep{SarndalSwenssonWretman1992,RobinsRotnitzkyZhao1994,Tsiatis2006,TaoZengLin2017,TaoZengLin2020}. In randomized trials, \citet{GilbertYuRotnitzky2014} derived optimal sampling probabilities when an expensive outcome is measured in a subsample and auxiliary variables are available for everyone. The optimal probabilities depend on the conditional residual variance of the outcome, which is rarely known when a trial starts. Adaptive, or multiwave, two-phase designs therefore estimate it from an initial set of measurements and revise the design as more measurements become available \citep{McIsaacCook2015,ChenLumley2020,SauerHedtGauthierHaneuse2023,ZengEtAl2025}. \citet{ZengEtAl2025} developed adaptive stratified designs for average causal effects in observational studies with costly confounders. \citet{MozerPashleyMiratrix2026} derived optimal stratified allocations for model-assisted estimation of treatment effects in randomized trials whose outcomes are scored by hand in a subsample. Their strata are formed within treatment arms, and the allocation is fixed in advance.

A related literature on prediction-powered inference combines machine-learning predictions with a small number of measured outcomes \citep{AngelopoulosEtAl2023}. Active statistical inference \citep{ZrnicCandes2024} uses the predictions to decide which outcomes to measure, corrects for the selection by inverse probability weighting, and allows the sampling rule to be updated as outcomes arrive, with inference from a martingale central limit theorem. \citet{KlugerBates2026} extended this approach to M-estimation under two-phase multiwave sampling. For sequential mean estimation, \citet{SfyrakiWang2026} found that mixing an uncertainty-based rule with a constant sampling probability gave the narrowest intervals when the weight on the constant component was near one.

In a blinded trial, the randomization code is ordinarily withheld from sponsor staff involved in running the trial \citep{ICHE91999}, so the statistician who manages outcome collection usually does not have it. We therefore study a design in which treatment assignment is used neither to fit the outcome model nor to set the sampling probabilities. The procedure uses pooled data from both arms, and the arms are compared only after unblinding, with the recorded sampling probabilities. Blinded sample size re-estimation works under the same restriction. There the outcome variance is estimated from pooled data and is inflated by a term proportional to the squared treatment effect \citep{GouldShih1992,KieserFriede2003}. A term of the same form appears in our comparison of blinded and unblinded designs.

The analysis uses a prediction for every participant and an inverse probability correction for participants whose outcomes are measured. This is the augmented inverse probability weighted estimator with known probabilities \citep{RobinsRotnitzkyZhao1994}, or the difference estimator of survey sampling \citep{SarndalSwenssonWretman1992}. The prediction model can be simple or flexible and can be refitted as outcomes accumulate. Validity does not require a correct model, because the correction uses the known sampling probabilities. The working models affect precision, and they direct measurements toward participants whose outcomes are hard to predict. The estimator and its martingale limit theory are close to those of \citet{ZrnicCandes2024}, and we restate them in Section~\ref{sec:estimation} for a difference of two arm means with fixed arm totals, with inference conditional on the complete trial data.

Our contributions concern the design. We bound the variance lost by learning the sampling rule while outcomes are being collected. The bound separates the error of the outcome prediction from the error of the residual variance model. It is linear in the latter in general and quadratic when the optimal probabilities are not truncated, and it relates the size of the initial random sample to the rate at which the models are learned. We then compare the blinded design with hypothetical designs that use treatment assignment. The difference splits into a term due to unequal residual variances in the two arms and a term due to the conditional treatment effect. The term due to the treatment effect is small near the null, while the term due to unequal residual variances does not depend on the treatment effect and can remain under the null. We also give a planning result based on classical Neyman allocation \citep{Neyman1934,Cochran1977}. With a known residual variance and no truncation, the optimal design measures a fraction $1/(1+\kappa^2)$ of the outcomes that simple random sampling needs for the same precision, where $\kappa$ is the coefficient of variation of the residual standard deviation.

We take the marginal mean difference as the estimand. The same analysis covers continuous and binary outcomes. The auxiliary variables may include baseline variables and measurements made after randomization but before the costly outcome is selected for measurement. Postrandomization variables are used only to predict the costly outcome and to guide sampling, and comparisons conditional on them are not given a causal interpretation. The data illustration uses one update after the initial random sample, because this design is easy to describe and audit. The theory allows repeated updates, and the simulations include a design with two updates.

The remainder of the paper is organized as follows. Section~\ref{sec:design} describes the design and Section~\ref{sec:estimation} the estimator. Section~\ref{sec:choosing} studies the choice of sampling probabilities, the cost of learning them, the size of the initial sample, and the cost of withholding treatment assignment, and it ends with a summary of the procedure. Sections~\ref{sec:simulation} and \ref{sec:toenail} report simulations and a resampling study of an antifungal trial. Section~\ref{sec:discussion} discusses limitations and extensions. Technical proofs are in Appendix~\ref{app:proofs}.

\section{Trial and sampling design}
\label{sec:design}

For participant $i=1,\ldots,n$, let $A_i\in\{0,1\}$ denote the randomized treatment, $X_i$ the auxiliary variables available before the costly outcome is selected for measurement, and $Y_i$ the costly outcome. Let $Y_i^a$ be the potential outcome under treatment $a$. Participants are independent and identically distributed apart from treatment assignment, which is independent of baseline variables and potential outcomes. The arm totals may be fixed by design. Under consistency and randomization,
\[
\mu_a=\E(Y_i^a)=\E(Y_i\mid A_i=a),
\qquad a\in\{0,1\},
\]
and the estimand is
\[
\Delta=\mu_1-\mu_0.
\]
Write $p_a=\Prb(A_i=a)$. The randomization probabilities are known and bounded away from zero.

Let $R_i=1$ if $Y_i$ is measured and $R_i=0$ otherwise, and let $\mathcal D_n=\{(A_i,X_i,Y_i):i=1,\ldots,n\}$ denote the complete trial data. The design starts with an initial simple random sample $S_0$ of $m$ participants, selected independently of $\mathcal D_n$, whose outcomes are all measured. Let $\mathcal M$ denote the other $n-m$ participants. They are indexed in the order in which their sampling decisions are made, and participants whose decisions are made at the same time are ordered arbitrarily. Let $\mathcal H_{i-1}$ denote the information available to the statistician before the decision for participant $i$. It may contain the auxiliary variables, $S_0$, earlier sampling decisions, and measured outcomes, but not the treatment assignments. The working prediction $f_i$ and the sampling probability $\pi_i$ are functions determined by $\mathcal H_{i-1}$, and we call such sequences predictable. We require
\begin{equation}
\Prb(R_i=1\mid \mathcal D_n,S_0,R_1,\ldots,R_{i-1})
=
\pi_i(X_i),
\qquad
\epsilon\leq\pi_i(X_i)\leq1.
\label{eq:sampling}
\end{equation}
Condition \eqref{eq:sampling} means that, given the information available to the statistician, the sampling draw does not depend on treatment assignments or on unmeasured outcomes, including those of other participants. It holds when the draws use random numbers generated independently of the trial data. Draws made at the same time are independent, so the number of measured outcomes is random with a prespecified expectation, as in Poisson sampling \citep{SarndalSwenssonWretman1992}. An outcome that must be collected for clinical or operational reasons can be given $\pi_i=1$. All results below are conditional on $S_0$, and participants in $S_0$ have $R_i=\pi_i=1$.

The sampling rule can be updated at prespecified times. The outcomes measured in the initial sample are used to fit a pooled outcome model and a model for the residual variance, and these fits set the sampling probabilities for the next group of participants. At a further update, both models are refitted to all outcomes measured so far. The probability applied to each participant is recorded and used in the analysis. Algorithm~\ref{alg:design} in Section~\ref{sec:practice} lists the steps.

As an example, consider a trial of 1000 participants randomized 1:1, with a budget of about 400 costly measurements. The auxiliary variables $X$ might include baseline variables, routine laboratory results, an inexpensive assay, a short-term outcome, or an imaging summary, all available for everyone. First, 150 participants are selected at random and their outcomes are measured. This sample contains participants from both arms, but the statistician who fits the sampling rule does not know who received which treatment. The 150 outcomes are used to estimate
\[
\widehat f_1(x)\approx \E(Y\mid X=x)
\qquad\text{and}\qquad
\widehat r_1(x)\approx
\E\left[\{Y-\widehat f_1(X)\}^2\mid X=x\right].
\]
The remaining 850 participants are then sampled with average probability $250/850\approx0.29$. Participants whose outcomes are hard to predict from $X$ receive higher probabilities. A participant with an uncertain prediction might be sampled with probability 0.80 and a participant with a precise prediction with probability 0.20. At database lock, the treatment assignments are unblinded and the arms are compared with the estimator of Section~\ref{sec:estimation}, using the recorded probabilities.

If $X$ contains postrandomization measurements, these can carry indirect information about treatment. This does not affect validity. It does mean that withholding the treatment code from the sampling procedure is a weaker condition than requiring all data used for sampling to be unrelated to treatment.

\section{Treatment effect estimation}
\label{sec:estimation}

For participant $i$ define the corrected outcome
\begin{equation}
U_i
=
f_i(X_i)
+
\frac{R_i}{\pi_i(X_i)}\{Y_i-f_i(X_i)\}.
\label{eq:corrected}
\end{equation}
For participants in $S_0$, $U_i=Y_i$. Let $n_a=\sum_{i=1}^n\ind(A_i=a)$. The arm means and the treatment effect are estimated by
\begin{equation}
\widehat\mu_a
=
\frac{1}{n_a}\sum_{i:A_i=a}U_i,
\qquad
\widehat\Delta=\widehat\mu_1-\widehat\mu_0.
\label{eq:estimator}
\end{equation}
The first term of \eqref{eq:corrected} is a prediction for every participant. The second term corrects the prediction error for participants whose outcomes are measured, and the correction is larger when the sampling probability is smaller. Let
\[
\Delta_n
=
\frac{1}{n_1}\sum_{i:A_i=1}Y_i
-
\frac{1}{n_0}\sum_{i:A_i=0}Y_i
\]
be the treatment difference that would be computed if every outcome were measured.

\begin{theorem}[Sampling unbiasedness]
\label{thm:unbiased}
If \eqref{eq:sampling} holds, then
\[
\E(\widehat\Delta\mid\mathcal D_n,S_0)=\Delta_n.
\]
The result holds for any predictable sequences $f_i$ and $\pi_i$ that satisfy \eqref{eq:sampling}, including sequences updated repeatedly from earlier measured outcomes.
\end{theorem}

Theorem~\ref{thm:unbiased} follows from the probability sampling design. The working prediction need not be a correct conditional mean, and the rule may change at every update. A poor model reduces precision but does not bias the estimator, provided the probabilities actually used are recorded. The estimator is the difference, or model-assisted, estimator of survey sampling \citep{SarndalSwenssonWretman1992,BreidtOpsomer2017} and the augmented inverse probability weighted estimator with known probabilities \citep{RobinsRotnitzkyZhao1994}. Theorem~\ref{thm:unbiased} also holds if $\pi_i$ depended on $A_i$. Treatment assignment is withheld to protect the blinding of trial conduct, and Section~\ref{sec:blinding-cost} quantifies the resulting loss of precision.

For large-sample inference, define
\[
h_{i,n}
=
\frac{n\ind(A_i=1)}{n_1}
-
\frac{n\ind(A_i=0)}{n_0}.
\]
The estimator has the exact decomposition
\begin{equation}
\sqrt n(\widehat\Delta-\Delta)
=
\sqrt n(\Delta_n-\Delta)
+
\frac{1}{\sqrt n}\sum_{i=1}^n
h_{i,n}
\left(\frac{R_i}{\pi_i}-1\right)
\{Y_i-f_i(X_i)\}.
\label{eq:decomp-estimator}
\end{equation}
The first term is the variation of the treatment comparison with all outcomes measured. Given $\mathcal D_n$ and $S_0$, the second term is a sum of martingale differences generated by the sampling design.

\begin{theorem}[Large-sample distribution]
\label{thm:clt}
Assume that
\[
\sqrt n(\Delta_n-\Delta)\overset{d}{\longrightarrow}N(0,V_Y),
\qquad V_Y=\frac{\Var(Y\mid A=1)}{p_1}+\frac{\Var(Y\mid A=0)}{p_0}.
\]
Suppose that \eqref{eq:sampling} holds with predictable $f_i$ and $\pi_i$, that $\sup_{i,n}\E\{Y_i^4+f_i(X_i)^4\}<\infty$, and that
\begin{equation}
\frac{1}{n}\sum_{i=1}^n
h_{i,n}^2
\frac{1-\pi_i}{\pi_i}
\{Y_i-f_i(X_i)\}^2
\overset{p}{\longrightarrow}V_R
\label{eq:predvar}
\end{equation}
for a constant $V_R<\infty$. Then
\[
\sqrt n(\widehat\Delta-\Delta)
\overset{d}{\longrightarrow}
N(0,V_Y+V_R).
\]
In addition, $n\,\widehat{\operatorname{se}}(\widehat\Delta)^2\overset{p}{\longrightarrow}V_Y+V_R$, where
\begin{equation}
\widehat{\operatorname{se}}(\widehat\Delta)
=
\left(\frac{s_1^2}{n_1}+\frac{s_0^2}{n_0}\right)^{1/2}
\label{eq:se}
\end{equation}
and $s_a^2$ is the sample variance of $U_i$ in arm $a$.
\end{theorem}

The theorem allows the prediction model and the sampling rule to be refitted repeatedly. Each fitted rule must be fixed before it is applied to a new group of participants. This holds when the rule is fitted to the outcomes measured up to a prespecified update and applied only afterward. The within-arm centering in \eqref{eq:se} suits the fixed arm totals common in trials. The fourth-moment condition implies the conditional Lindeberg condition of the martingale central limit theorem. Condition \eqref{eq:predvar} holds when the fitted rules settle down, for example when $f_i$ and $\pi_i$ converge to fixed functions $\bar f$ and $\bar\pi$. Under equal allocation $V_R$ is then $4L_{r_{\bar f}}(\bar\pi)$ in the notation of \eqref{eq:variance} below. Theorem~\ref{thm:unbiased} corresponds to the martingale-difference property used by \citet{ZrnicCandes2024}, and Theorem~\ref{thm:clt} adapts their Proposition~6.1 (as in the proceedings version) to the present setting.

\begin{remark}[Arm-specific predictions at the final analysis]
\label{rem:arm-specific}
Theorem~\ref{thm:unbiased} requires $f_i$ to be fixed before $R_i$ is drawn. After unblinding, one may wish to replace the pooled prediction in \eqref{eq:corrected} by arm-specific predictions fitted to the measured outcomes. With the true arm-specific means, this can only reduce the sampling part of the variance for a given design (Section~\ref{sec:blinding-cost}). With fitted predictions, exact unbiasedness no longer follows from Theorem~\ref{thm:unbiased}, and large-sample validity requires further conditions, for example cross-fitting \citep{ChernozhukovEtAl2018}. We do not study this estimator here.
\end{remark}

\section{Choosing the sampling probabilities}
\label{sec:choosing}

\subsection{Variance for a fixed rule}

The calculation is simplest under equal allocation with a common working prediction $f(X)$ and sampling probability $\pi(X)$ that do not change during the trial. Let
\[
r_f(x)=\E\{(Y-f(X))^2\mid X=x\}.
\]
The asymptotic variance of \eqref{eq:estimator} can be written as
\begin{equation}
V(f,\pi)
=
V_{\mathrm{full}}
+
4L_{r_f}(\pi),
\qquad
L_r(\pi)
=
\E\left\{\frac{1-\pi(X)}{\pi(X)}r(X)\right\},
\label{eq:variance}
\end{equation}
where $V_{\mathrm{full}}=2\Var(Y\mid A=1)+2\Var(Y\mid A=0)$ is the variance when every outcome is measured. Only $L_{r_f}(\pi)$ depends on the sampling rule.

\begin{remark}[Unequal allocation]
\label{rem:unequal}
For general $p_1$ the sampling part of the variance is
\[
\E\left[\frac{1-\pi(X)}{\pi(X)}\,\frac{\{Y-f(X)\}^2}{p_A^2}\right]
\]
(Appendix~\ref{app:proofs}). With baseline $X$, its conditional version is $\sum_a r_{f,a}(x)/p_a$, where $r_{f,a}(x)=\E[\{Y-f(X)\}^2\mid A=a,X=x]$. This quantity weights the two arms differently and cannot be estimated from pooled data. The pooled residual variance used below is therefore the right target only under equal allocation. Under unequal allocation the blinded rule minimizes a different objective, and the resulting loss is similar in kind to $G_\sigma$ in Theorem~\ref{thm:pooled-gap}.
\end{remark}

For a target fraction $\rho$ of outcomes measured and bounds $\epsilon<\rho\leq u\leq1$, let
\[
\mathcal P_\rho
=
\{\pi:\epsilon\leq\pi(X)\leq u,\ \E\pi(X)=\rho\}.
\]
For a nonnegative function $r$, write $\pi_r^*$ for a minimizer of $L_r(\pi)$ over $\mathcal P_\rho$.

\begin{proposition}[Oracle sampling probability]
\label{prop:optimal}
For fixed $f$ with $\Prb\{r_f(X)>0\}=1$, $\pi_{r_f}^*$ has the form
\begin{equation}
\pi_{r_f}^*(x)
=
\min\left[u,
\max\left\{\epsilon,c\sqrt{r_f(x)}\right\}
\right],
\label{eq:optimal}
\end{equation}
where $c>0$ is chosen so that $\E\pi(X)=\rho$. If neither bound is active,
\[
\pi_{r_f}^*(x)
=
\rho\frac{\sqrt{r_f(x)}}{\E\{\sqrt{r_f(X)}\}}.
\]
\end{proposition}

This is Neyman allocation \citep{Neyman1934,Cochran1977} applied to the prediction residuals. The same rule gives the optimal two-phase sampling probabilities for this estimator in randomized trials \citep{GilbertYuRotnitzky2014} and for design-based estimators \citep{ChenLumley2022}, and \citet{ZrnicCandes2024} derived it for active inference. We use it as an oracle benchmark. In practice $r_f$ is unknown and changes when the working prediction changes, so an adaptive design must learn the prediction and its residual variance while outcomes are being collected.

The oracle rule also gives a simple planning calculation. Let $K_f=\E\sqrt{r_f(X)}$ and let $\kappa_f=[\Var\{\sqrt{r_f(X)}\}]^{1/2}/K_f$ be the coefficient of variation of the residual standard deviation.

\begin{proposition}[Measurements saved by the oracle rule]
\label{prop:saving}
Assume equal allocation, fix $f$, and suppose that the oracle probabilities are not truncated.
\begin{enumerate}[label=(\roman*),nosep,leftmargin=2em]
\item Suppose $\rho(1+\kappa_f^2)\leq1$. Simple random sampling of a fraction $\rho_s$, with the same prediction $f$, has the same asymptotic variance as the oracle rule with fraction $\rho$ if and only if $\rho_s=\rho(1+\kappa_f^2)$.
\item The oracle rule attains asymptotic variance $v$, on the scale of $n\Var(\widehat\Delta)$, when
\[
\rho=\frac{4K_f^2}{v-V_{\mathrm{full}}+4\E r_f(X)},
\]
provided that this value lies in $(0,1]$ and the resulting probabilities are not truncated.
\end{enumerate}
\end{proposition}

Part (i) is the classical comparison of Neyman allocation with proportional allocation \citep{Cochran1977}. For the same precision, the oracle rule measures a fraction $1/(1+\kappa_f^2)$ of the outcomes that simple random sampling measures. The saving depends only on how much the residual standard deviation varies across participants, and it is zero when the residual variance is constant. Part (ii) can be used at the planning stage, with values of $K_f$, $\E r_f(X)$, and $V_{\mathrm{full}}$ taken from earlier trials or pilot data. Both parts describe an oracle. The adaptive design also pays for learning the rule and for its initial sample, and the next two subsections quantify these costs.

\subsection{Cost of learning the sampling rule}

The calculations below use population expectations for readability. The same inequalities hold with empirical averages over a prespecified group of participants whose auxiliary variables are already available. Suppose that before sampling stage $k$ the outcomes already measured are used to construct a working prediction $\widehat f_k$ and an estimate $\widehat r_k$ of its conditional residual variance. Let $\mathcal G_k$ denote the information available before stage $k$. Because participants are independent and $S_0$ is chosen at random, the conditional distribution of $Y$ given $X$ among stage-$k$ participants does not depend on $\mathcal G_k$. Define
\[
r_k(x)
=
\E\{(Y-\widehat f_k(X))^2\mid X=x,\mathcal G_k\}
\qquad\text{and}\qquad
\widehat\pi_k
\in
\arg\min_{\pi\in\mathcal P_\rho}L_{\widehat r_k}(\pi).
\]
The next result gives the price of using the estimated rule instead of the oracle rule for the current working prediction. Let
\[
C_\epsilon=\frac{1-\epsilon}{\epsilon},
\]
and write $\|g\|_p=\{\E|g(X)|^p\}^{1/p}$, with the expectation taken over the auxiliary variables of stage-$k$ participants.

\begin{theorem}[Error from estimating the sampling rule]
\label{thm:adaptive-bound}
Conditional on $\mathcal G_k$, the following hold.
\begin{enumerate}[label=(\roman*),nosep,leftmargin=2em]
\item For every estimate $\widehat r_k$,
\begin{equation}
0\leq
L_{r_k}(\widehat\pi_k)
-
L_{r_k}(\pi_{r_k}^*)
\leq
\E\left[\left\{\frac{1}{\widehat\pi_k(X)}-\frac{1}{\pi_{r_k}^*(X)}\right\}\{r_k(X)-\widehat r_k(X)\}\right]
\leq
C_\epsilon\|\widehat r_k-r_k\|_1.
\label{eq:regret1}
\end{equation}
\item If neither $\widehat\pi_k$ nor $\pi^*_{r_k}$ is truncated and $\widehat r_k>0$, then
\begin{equation}
L_{r_k}(\widehat\pi_k)
-
L_{r_k}(\pi_{r_k}^*)
\leq
\frac{\E\sqrt{\widehat r_k(X)}}{\rho}\,
\E\left[\frac{\{\sqrt{r_k(X)}-\sqrt{\widehat r_k(X)}\}^2}{\sqrt{\widehat r_k(X)}}\right].
\label{eq:regret-quad}
\end{equation}
In particular, if $\widehat r_k\geq r_{\min}>0$, the right-hand side is at most
\[
\frac{\E\sqrt{\widehat r_k(X)}}{\rho\sqrt{r_{\min}}}\,\|\sqrt{\widehat r_k}-\sqrt{r_k}\|_2^2.
\]
\item Let
\[
f_P(x)=\E(Y\mid X=x),
\qquad
r_P(x)=\E\{(Y-f_P(X))^2\mid X=x\},
\]
and let $\Lambda_k$ be the right-hand side of \eqref{eq:regret1} or, when (ii) applies, the smaller of the right-hand sides of \eqref{eq:regret1} and \eqref{eq:regret-quad}. Then
\begin{equation}
0\leq
L_{r_k}(\widehat\pi_k)
-
L_{r_P}(\pi_{r_P}^*)
\leq
\Lambda_k
+
C_\epsilon\|\widehat f_k-f_P\|_2^2.
\label{eq:regret2}
\end{equation}
\end{enumerate}
The extra contribution to the variance during stage $k$ is therefore at most four times the right-hand side of \eqref{eq:regret2}.
\end{theorem}

The middle term in \eqref{eq:regret1} averages the product of two errors. One is the error of the fitted residual variance, and the other is the resulting error of the inverse sampling probabilities. The linear bound uses only that both probabilities lie in $[\epsilon,u]$. When the map from $r$ to $\pi_r^*$ is Lipschitz, the product is of order $\|\widehat r_k-r_k\|_2^2$, and part (ii) makes this explicit without truncation. The excess is second order because first-order errors in the probabilities do not change the objective at its minimum. The prediction term in \eqref{eq:regret2} is also quadratic. A good predictor combined with a poor residual variance model can still sample inefficiently, and a good residual variance model does not remove the prediction term. Neither error affects validity. Both determine how close the design comes to the oracle variance.

If the rule is updated several times, the same argument applies at each stage. Suppose there are $N_n$ sampling stages, and let $w_{k,n}$ be the proportion of participants assigned to stage $k$, excluding an initial random sample whose fraction is $o(1)$. Define the average sampling contribution and its oracle value,
\[
\overline V_{R,n}
=
4\sum_{k=1}^{N_n}w_{k,n}L_{r_k}(\widehat\pi_k)
\qquad\text{and}\qquad
V_R^*=4L_{r_P}(\pi_{r_P}^*).
\]

\begin{corollary}[Efficiency under repeated updating]
\label{cor:repeated}
Under the conditions of Theorem~\ref{thm:adaptive-bound},
\begin{equation}
0\leq
\overline V_{R,n}-V_R^*
\leq
4
\sum_{k=1}^{N_n}w_{k,n}
\left
[\Lambda_k+C_\epsilon\|\widehat f_k-f_P\|_2^2
\right]
+o_p(1),
\label{eq:repeated-bound}
\end{equation}
where the $o_p(1)$ term accounts for an asymptotically negligible initial random sample and for empirical normalization of the stage-specific probabilities. In particular, $\Lambda_k\leq C_\epsilon\|\widehat r_k-r_k\|_1$. If the weighted errors on the right-hand side converge to zero, the adaptive design has the same first-order sampling variance as the pooled oracle. Combined with Theorem~\ref{thm:clt}, its treatment effect estimator then has the oracle first-order variance.
\end{corollary}

Inequality~\eqref{eq:repeated-bound} is useful even when the oracle approximation is not yet accurate. It shows how a small initial sample, an unstable residual variance model, or repeated refitting with little new information can leave a noticeable efficiency gap. More frequent updating does not automatically improve precision. It helps only when the added outcomes reduce the weighted prediction and residual variance errors. If the initial random sample remains a fixed fraction of the trial as $n$ increases, its contribution does not vanish, and the design cannot match an oracle that allocates the whole budget from the start. The lower bound in \eqref{eq:repeated-bound} also holds if the stages use different expected fractions $\rho_k$ with $\sum_kw_{k,n}\rho_k=\rho$, because the minimum of $L_{r_P}$ over $\mathcal P_\rho$ is a convex function of $\rho$.

\subsection{Size of the initial random sample}

The bound also indicates how large the initial random sample should be. Suppose the design measures a fraction $\rho$ of the outcomes in expectation and first measures the outcomes of $m_n$ randomly selected participants before fitting the adaptive rule. Write $\alpha_n=m_n/n$. The remaining participants are then sampled with average probability
\[
\rho_n=\frac{\rho-\alpha_n}{1-\alpha_n}.
\]
A larger initial sample improves the fitted rule but spends more budget before the rule is available.

\begin{corollary}[Size of the initial random sample]
\label{cor:initial-size}
Assume $m_n\to\infty$ and $m_n/n\to0$, and that the oracle probabilities are not truncated for sampling fractions in a neighborhood of $\rho$. Let $\widehat f_1$ and $\widehat r_1$ be fitted to the $m_n$ initial outcomes, let $r_1$ be defined as in Theorem~\ref{thm:adaptive-bound}, and suppose that
\begin{equation}
\Lambda_1+C_\epsilon\|\widehat f_1-f_P\|_2^2
=
O_p(m_n^{-\gamma})
\label{eq:learning-rate}
\end{equation}
for some $\gamma>0$. Then the excess sampling variance relative to the pooled oracle with the same overall fraction $\rho$ is
\begin{equation}
O_p\left(\frac{m_n}{n}+m_n^{-\gamma}\right),
\label{eq:initial-regret}
\end{equation}
and the $m_n/n$ term has leading constant $4[K^2(1-\rho)^2/\rho^2+\Var\{\sqrt{r_P(X)}\}]$, where $K=\E\sqrt{r_P(X)}$. The order of this bound is smallest when
\[
m_n\asymp n^{1/(1+\gamma)},
\]
and the excess variance is then of order $n^{-\gamma/(1+\gamma)}$. For example, suppose both working models are correctly specified parametric models, so that $\|\widehat r_1-r_1\|_1$, $\|\sqrt{\widehat r_1}-\sqrt{r_1}\|_2$, and $\|\widehat f_1-f_P\|_2$ are $O_p(m_n^{-1/2})$. The linear bound in \eqref{eq:regret1} gives $\gamma=1/2$ and $m_n\asymp n^{2/3}$. If $r_P$ is bounded away from zero and truncation is inactive, \eqref{eq:regret-quad} gives $\gamma=1$, so $m_n\asymp n^{1/2}$ and the excess variance is of order $n^{-1/2}$.
\end{corollary}

The rate describes a tradeoff. In a particular trial, constants, truncation, and model complexity also matter. An initial sample that is too small leaves the fitted rule uncertain. An initial sample that stays a fixed fraction of a growing trial spends a nonvanishing share of the budget before the rule can adapt. The cost of the initial sample is positive even when $r_P$ is constant, because initial participants are measured with probability one instead of $\rho$. The exponent depends on which bound in Theorem~\ref{thm:adaptive-bound} applies, and the quadratic bound leads to a smaller initial sample than the linear bound. Earlier work on multiwave designs reported this tradeoff empirically \citep{McIsaacCook2015,ChenLumley2020}. Corollary~\ref{cor:initial-size} expresses it through the learning rate of the models that set the sampling probabilities.

\subsection{Cost of withholding treatment assignment}
\label{sec:blinding-cost}

The proposed procedure uses treatment assignment neither for the outcome model nor for the sampling probabilities. To measure the cost of this restriction, we compare it with hypothetical oracle designs that may use the randomization code. Let
\[
q_a(x)=\E(Y\mid A=a,X=x),
\qquad
\sigma_a^2(x)=\Var(Y\mid A=a,X=x),
\]
and write
\[
\eta(x)=\Prb(A=1\mid X=x),
\qquad
d(x)=q_1(x)-q_0(x),
\qquad
s(x)=\eta(x)\sigma_1^2(x)+\{1-\eta(x)\}\sigma_0^2(x).
\]
The pooled conditional mean is
\[
f_P(x)
=
\eta(x)q_1(x)+\{1-\eta(x)\}q_0(x),
\]
and its conditional residual variance is
\begin{equation}
r_P(x)
=
s(x)+\eta(x)\{1-\eta(x)\}d^2(x).
\label{eq:decomposition}
\end{equation}
The first term averages the residual variances of the two arms. The second appears because a pooled prediction does not use treatment assignment. If $X$ contains only baseline variables in a 1:1 randomized trial, $\eta(x)=1/2$, and with a constant treatment effect $\Delta$ the second term is $\Delta^2/4$. The pooled variance used in blinded sample size re-estimation is inflated by the same term \citep{GouldShih1992,KieserFriede2003}.

A second cost arises even if the arm-specific means $q_a$ were known. The best sampling probabilities differ between arms when $\sigma_1(x)$ and $\sigma_0(x)$ differ, and a blinded rule must use one probability for both. The next result uses two hypothetical oracle designs to separate these sources of loss. For closed-form expressions, we assume equal allocation, baseline $X$, and probabilities away from the truncation limits.

\begin{theorem}[Cost of keeping treatment assignment out of the sampling rule]
\label{thm:pooled-gap}
Assume that $A$ is independent of baseline $X$ with $\Prb(A=1)=1/2$, and that all oracle probabilities lie strictly between $\epsilon$ and $u$. All designs measure the same expected fraction $\rho$ of outcomes. When the sampling probabilities may depend on treatment, this means $\{\E\pi_1(X)+\E\pi_0(X)\}/2=\rho$. Let $V_T^*$ be the minimum asymptotic variance when treatment assignment may be used in both the outcome model and the sampling rule. Let $V_S^*$ be the minimum when the arm-specific means $q_a$ are used but one probability $\pi(X)$ is required for both arms. Let $V_P^*$ be the minimum when both the outcome model and the sampling rule are pooled over treatment, as in the proposed procedure.

With
\[
s(x)=\frac{\sigma_1^2(x)+\sigma_0^2(x)}{2},
\]
the loss from requiring one sampling probability is
\begin{equation}
G_\sigma(\rho)
=
V_S^*-V_T^*
=
\frac{
4\{\E\sqrt{s(X)}\}^2
-
\{\E\sigma_1(X)+\E\sigma_0(X)\}^2
}{\rho}
\geq0,
\label{eq:gsigma}
\end{equation}
with equality if and only if $\sigma_1(X)=\sigma_0(X)$ almost surely. If $\sigma_1$ and $\sigma_0$ are constants, $G_\sigma(\rho)=(\sigma_1-\sigma_0)^2/\rho$. The further loss from omitting treatment from the outcome prediction is
\begin{align}
G_d(\rho)
=
V_P^*-V_S^*
&=
4\left\{
\min_{\pi\in\mathcal P_\rho}L_{r_P}(\pi)
-
\min_{\pi\in\mathcal P_\rho}L_s(\pi)
\right\} \nonumber\\
&=\frac{4}{\rho}\left[\{\E\sqrt{r_P(X)}\}^2-\{\E\sqrt{s(X)}\}^2\right]-\E\{d^2(X)\}
\geq0,
\label{eq:gd}
\end{align}
where $r_P(x)=s(x)+d^2(x)/4$ under equal allocation. The loss $G_d$ also satisfies
\begin{equation}
G_d(\rho)\leq\E\left[\frac{1-\pi_s^*(X)}{\pi_s^*(X)}\,d^2(X)\right]\leq C_\epsilon\|d\|_2^2,
\label{eq:gd-bound}
\end{equation}
where $\pi_s^*$ is the oracle probability for $s$. If $s(x)\geq s_{\min}>0$, the first inequality in \eqref{eq:gd-bound} is an equality up to a term of order $\|d\|_\infty^4$. The total loss from withholding treatment assignment is
\begin{equation}
V_P^*-V_T^*
=
G_\sigma(\rho)+G_d(\rho).
\label{eq:gap-decomp}
\end{equation}
Under local alternatives with $\|d_n\|_2=O(n^{-1/2})$, $G_d(\rho)=O(n^{-1})$. If in addition $\sigma_1(X)=\sigma_0(X)$ almost surely, the pooled oracle and the oracle that uses treatment assignment have the same first-order variance.
\end{theorem}

When $\sigma_1$ and $\sigma_0$ are constants, the oracle that uses treatment assignment is Neyman allocation between the arms, $\pi_a\propto\sigma_a$, and $G_\sigma=(\sigma_1-\sigma_0)^2/\rho$. For example, with $\sigma_1=1.5\sigma_0$, $\rho=0.4$, and conditional means that do not depend on $X$, $G_\sigma=0.625\sigma_0^2$, about 4\% of $V_T^*=15.6\sigma_0^2$. The loss from a common sampling probability is therefore small unless the residual standard deviations differ markedly between arms.

\begin{corollary}[Overall comparison with an oracle that uses treatment assignment]
\label{cor:combined-gap}
Under the conditions of Corollary~\ref{cor:repeated} and Theorem~\ref{thm:pooled-gap}, let
$\overline V_n^{\mathrm{ad}}
=
V_{\mathrm{full}}+\overline V_{R,n}.$
Then
\begin{align}
0\leq
\overline V_n^{\mathrm{ad}}-V_T^*
\leq{}&
4
\sum_{k=1}^{N_n}w_{k,n}
\left[
\Lambda_k+C_\epsilon\|\widehat f_k-f_P\|_2^2
\right] \nonumber\\
&+G_\sigma(\rho)+G_d(\rho)+o_p(1).
\label{eq:combined-gap}
\end{align}
If the weighted learning errors converge to zero, $\|d_n\|_2=O(n^{-1/2})$, and $\sigma_1(X)=\sigma_0(X)$ almost surely, the adaptive design and the oracle that uses treatment assignment have the same first-order variance.
\end{corollary}

The decomposition separates two consequences of withholding treatment assignment. The term $G_\sigma$ arises because one sampling rule cannot reproduce different optimal probabilities in the two arms when their residual variances differ. It does not depend on $d$ and can remain under the null. The term $G_d$ arises because a pooled outcome model absorbs part of the conditional treatment effect into its residual variance. It is of order $\|d\|_2^2$ and is negligible when the conditional treatment effect is small. The learning term in \eqref{eq:combined-gap} decreases as more outcomes are measured, and the two blinding terms do not, because they come from the decision to withhold treatment assignment.

\subsection{Summary of the procedure}
\label{sec:practice}

\begin{algorithm}[t]
\caption{Blinded adaptive sampling of a costly outcome}
\label{alg:design}
\begin{enumerate}[label=\arabic*.,leftmargin=2em,itemsep=2pt]
\item Before sampling starts, prespecify the auxiliary variables, the expected fraction $\rho$ of outcomes to be measured, the initial fraction $\alpha$, the probability bounds $[\epsilon,u]$, the working models, and the update times.
\item Select an initial simple random sample $S_0$ of $m=\alpha n$ participants and measure their outcomes.
\item Without the treatment code, fit the working prediction $\widehat f$ and the residual variance model $\widehat r$ to the measured outcomes. For a binary outcome, $\widehat r=\widehat f(1-\widehat f)$ may be used.
\item For the next group of participants, set $\pi_i=\min[u,\max\{\epsilon,c\,\widehat r(X_i)^{1/2}\}]$, with $c$ chosen so that the probabilities have the prespecified average. Draw $R_i$ from a Bernoulli$(\pi_i)$ distribution using random numbers generated independently of the trial data. Record $\pi_i$ and the model version.
\item At each prespecified update, refit $\widehat f$ and $\widehat r$ to all measured outcomes, weighting each by the inverse of its sampling probability, and return to step 4.
\item After database lock and unblinding, compute $U_i$ from \eqref{eq:corrected}, the estimate \eqref{eq:estimator}, the standard error \eqref{eq:se}, and a Wald interval.
\end{enumerate}
\end{algorithm}

Algorithm~\ref{alg:design} summarizes the procedure. A single update after the initial sample is attractive in a confirmatory trial because it is simple to implement and audit. More updates help when the added outcomes improve the fitted rule (Corollary~\ref{cor:repeated}). Proposition~\ref{prop:saving} gives the saving that an oracle rule would achieve, and the learning loss in Theorem~\ref{thm:adaptive-bound} reduces it. In our simulations with 40\% of outcomes measured, the adaptive design achieved between one third and two thirds of the oracle saving (Section~\ref{sec:sim-saving}). An initial sample of 5\% to 10\% worked well for trials of 2000 to 8000 participants, and 10\% to 20\% for a trial of 272 participants (Sections~\ref{sec:simulation} and \ref{sec:toenail}). For a binary outcome, the variance $\widehat p(1-\widehat p)$ implied by the outcome model was more stable than a separate residual variance model. If it simplifies operations, the sampling probabilities can be rounded to a few levels, for example three strata of predicted difficulty. Theorems~\ref{thm:unbiased} and \ref{thm:clt} apply to any predictable probabilities, so rounding does not affect validity. The efficiency results of Section~\ref{sec:choosing} assume the optimized probabilities.

\section{Simulation study}
\label{sec:simulation}

\subsection{Settings and designs}

Each simulated trial had $n=2000$ participants, exactly half assigned to each arm. The auxiliary variables were independent, with $X_1,X_2\sim N(0,1)$, $X_3\sim\mathrm{Bernoulli}(0.4)$, and $X_4\sim\mathrm{Uniform}(-1,1)$. Outcomes were generated as $Y=q_A(X)+\sigma_A(X)\zeta$, where $\zeta$ is a standardized $t_6$ variable independent of $(A,X)$. Let
\[
q(X)=0.75X_1+0.55(X_2^2-1)+0.45X_3+0.35\sin(\pi X_4).
\]
We used three settings.
\begin{enumerate}[label=\arabic*.,leftmargin=2em]
\item Heteroskedastic outcome. $q_0=q$, $q_1=q+0.30$, and $\sigma_0=\sigma_1=\exp(-0.4+0.7X_1+0.5X_3)$.
\item Heterogeneous treatment effects. $q_0=q$, $q_1=q+0.30+0.5X_1$, $\sigma_0=\exp(-0.4+0.7X_1+0.5X_3)$, and $\sigma_1=\exp(-0.4-0.7X_1+0.5X_3)$.
\item Misspecified working models. $q_0=q+0.6X_1X_2+0.5\{\cos(\pi X_1)-e^{-\pi^2/2}\}$, $q_1=q_0+0.30$, and $\sigma_0=\sigma_1=\exp\{-0.4+0.5X_1+0.6\tanh(X_1X_2)\}$.
\end{enumerate}
Since $\E\cos(\pi X_1)=e^{-\pi^2/2}$, the marginal treatment effect is 0.30 in every setting. In setting 2 the residual standard deviations of the two arms change with $X_1$ in opposite directions, so $G_\sigma>0$, and the conditional effect varies with $X_1$, so $G_d$ is larger than in the other settings. In settings 1 and 3, $G_\sigma=0$.

The outcome was first measured in a simple random sample of 10\% of participants ($m=200$). The expected fraction measured was 40\% in total, so the remaining participants were sampled with average probability $1/3$, and probabilities were restricted to $[0.05,0.90]$. The working outcome model was a ridge regression of $Y$ on standardized $X_1,\ldots,X_4,X_1^2,X_2^2$ with an unpenalized intercept. The penalty was chosen by exact leave-one-out cross-validation over 13 values from $10^{-3}$ to $10^3$. The working model omits the sine term, and in setting 3 also the interaction and cosine terms. For the residual variance model, the measured outcomes were split at random into two halves. The outcome model was fitted to each half, and residuals were computed on the other half. A ridge regression of the log squared residuals on the same variables, tuned in the same way, gave $\widehat r$ as the exponential of its fitted values. To avoid taking the logarithm of zero, $10^{-8}$ times the sample variance of $Y$ was added to the squared residuals. Sampling probabilities were proportional to $\widehat r^{1/2}$, with the constant $c$ in \eqref{eq:optimal} found by bisection so that the probabilities averaged $1/3$ over the remaining participants. Because the probabilities are normalized, a constant multiplicative bias in $\widehat r$, such as the bias from modeling log squared residuals, does not change the rule. In the design with two updates, the remaining participants were split at random into two halves. The first half was sampled with the rule from the initial sample. Both models were then refitted to all measured outcomes, weighted by the inverse of their sampling probabilities, and the refitted prediction and rule were used for the second half with the same average probability.

We compared these designs with measuring every outcome and with simple random sampling of the remaining participants, using the same fitted prediction. We also used four oracle rules based on the true model, each with the same initial sample and budget.
\begin{enumerate}[label=(\alph*),leftmargin=2em]
\item Oracle probabilities for the fitted prediction, $\pi\propto\{r_P(x)+(\widehat f_1(x)-f_P(x))^2\}^{1/2}$, which is $\pi^*_{r_1}$ in Theorem~\ref{thm:adaptive-bound}.
\item The pooled oracle, with $f_P$ in the estimator and $\pi\propto r_P^{1/2}$.
\item Arm-specific means $q_A$ in the estimator, with the common probability $\pi\propto s^{1/2}$.
\item Arm-specific means with arm-specific probabilities $\pi_a\propto\sigma_a$.
\end{enumerate}
Rules (b), (c), and (d) correspond to $V_P^*$, $V_S^*$, and $V_T^*$ in Theorem~\ref{thm:pooled-gap}. Moving from the adaptive design to rules (a), (b), (c), and (d) in turn removes the error of the residual variance model, the prediction error, the term $G_d$, and the term $G_\sigma$.

Each setting used 5000 replications. All designs were applied to the same simulated trials and used the same uniform random numbers for the sampling draws, so differences between designs have small Monte Carlo error. Standard errors were computed from \eqref{eq:se}. Relative efficiency was the variance of $\widehat\Delta$ under simple random sampling divided by its variance under the design, and its Monte Carlo standard error was obtained by resampling replications. The Monte Carlo standard error of coverage is about 0.003.

\subsection{Results with a 10\% initial sample}

Table~\ref{tab:simulation} shows negligible bias and coverage between 0.947 and 0.956 for every design. The measured fraction averaged 0.400, with a standard deviation of 0.009 to 0.010, about 20 participants. With one update, adaptive sampling had relative efficiency 1.34, 1.13, and 1.12 in the three settings. A second update raised these values to 1.40, 1.16, and 1.13.

Table~\ref{tab:decomp} splits the difference between the adaptive design and the oracle that uses treatment assignment into the terms of Corollary~\ref{cor:combined-gap}, on the scale of $n$ times the variance. In setting 1 the two learning terms account for most of the gap. Learning the residual variance cost 1.29 and prediction error cost 0.91, compared with 0.37 for $G_d$ and 0 for $G_\sigma$. Withholding treatment assignment therefore cost little, and the second update recovered most of the residual variance term, with relative efficiency 1.40 against 1.42 for the oracle probabilities. In setting 2 the largest term was $G_\sigma=3.19$, followed by learning the residual variance (2.34), prediction error (0.74), and $G_d=0.45$. In setting 3 the prediction term (3.64) was the largest, because the misspecified outcome model leaves $(\widehat f-f_P)^2$ in the residual (Theorem~\ref{thm:adaptive-bound}(iii)). Refitting the same model at a second update barely changed the efficiency. The simulated values of $G_d$ and $G_\sigma$ agree with the values computed from Theorem~\ref{thm:pooled-gap} for this design, and the largest difference, for $G_\sigma$ in setting 2, is 1.7 Monte Carlo standard errors.

\begin{table}[t]
\centering
\caption{Simulation results with 5000 replications per setting, $n=2000$, a 10\% initial sample, and 40\% of outcomes measured in expectation. Coverage refers to nominal 95\% Wald intervals, and its Monte Carlo standard error is about 0.003. RE is the relative efficiency against simple random sampling with the same initial sample, with its Monte Carlo standard error in parentheses. Rows below the rule in each setting are oracle designs that use the true model.}
\label{tab:simulation}
\small
\setlength{\tabcolsep}{3.3pt}
\begin{tabular}{lrrrrr}
\toprule
Design & Bias & Empirical SD & Mean SE & Coverage & RE\\
\midrule
\multicolumn{6}{l}{\textit{Setting 1 (heteroskedastic outcome)}}\\
\quad All outcomes measured & 0.000 & 0.080 & 0.081 & 0.952 & 2.27 (0.05)\\
\quad Simple random sampling & 0.002 & 0.121 & 0.120 & 0.952 & 1.00\\
\quad Adaptive, one update & 0.002 & 0.105 & 0.104 & 0.950 & 1.34 (0.02)\\
\quad Adaptive, two updates & 0.001 & 0.102 & 0.102 & 0.953 & 1.40 (0.02)\\
\cmidrule(l){1-6}
\quad Oracle probabilities, fitted prediction & 0.001 & 0.101 & 0.101 & 0.951 & 1.42 (0.03)\\
\quad Pooled oracle ($V_P^*$) & 0.001 & 0.099 & 0.099 & 0.949 & 1.49 (0.03)\\
\quad Arm-specific means, common probability ($V_S^*$) & 0.001 & 0.098 & 0.098 & 0.952 & 1.51 (0.03)\\
\quad Arm-specific means and probabilities ($V_T^*$) & 0.001 & 0.098 & 0.098 & 0.952 & 1.51 (0.03)\\
\addlinespace
\multicolumn{6}{l}{\textit{Setting 2 (heterogeneous treatment effects)}}\\
\quad All outcomes measured & $-$0.002 & 0.087 & 0.087 & 0.954 & 2.10 (0.04)\\
\quad Simple random sampling & $-$0.001 & 0.126 & 0.125 & 0.949 & 1.00\\
\quad Adaptive, one update & $-$0.001 & 0.118 & 0.118 & 0.949 & 1.13 (0.02)\\
\quad Adaptive, two updates & 0.000 & 0.117 & 0.118 & 0.954 & 1.16 (0.02)\\
\cmidrule(l){1-6}
\quad Oracle probabilities, fitted prediction & $-$0.001 & 0.113 & 0.114 & 0.955 & 1.24 (0.02)\\
\quad Pooled oracle ($V_P^*$) & $-$0.001 & 0.111 & 0.112 & 0.956 & 1.27 (0.02)\\
\quad Arm-specific means, common probability ($V_S^*$) & $-$0.001 & 0.110 & 0.111 & 0.955 & 1.30 (0.02)\\
\quad Arm-specific means and probabilities ($V_T^*$) & $-$0.001 & 0.103 & 0.103 & 0.951 & 1.49 (0.03)\\
\addlinespace
\multicolumn{6}{l}{\textit{Setting 3 (misspecified working models)}}\\
\quad All outcomes measured & 0.001 & 0.074 & 0.074 & 0.948 & 2.05 (0.04)\\
\quad Simple random sampling & 0.001 & 0.106 & 0.105 & 0.952 & 1.00\\
\quad Adaptive, one update & 0.002 & 0.100 & 0.100 & 0.951 & 1.12 (0.01)\\
\quad Adaptive, two updates & 0.002 & 0.100 & 0.099 & 0.951 & 1.13 (0.02)\\
\cmidrule(l){1-6}
\quad Oracle probabilities, fitted prediction & 0.002 & 0.096 & 0.095 & 0.947 & 1.21 (0.02)\\
\quad Pooled oracle ($V_P^*$) & 0.002 & 0.086 & 0.086 & 0.948 & 1.51 (0.03)\\
\quad Arm-specific means, common probability ($V_S^*$) & 0.002 & 0.085 & 0.085 & 0.947 & 1.54 (0.03)\\
\quad Arm-specific means and probabilities ($V_T^*$) & 0.002 & 0.085 & 0.085 & 0.947 & 1.54 (0.03)\\
\bottomrule
\end{tabular}
\end{table}

\begin{table}[t]
\centering
\caption{Decomposition of the variance gap between one-update adaptive sampling and the oracle that uses treatment assignment, in units of $n$ times the variance of $\widehat\Delta$, for the designs in Table~\ref{tab:simulation}. The residual variance term is the adaptive design minus the rule with oracle probabilities for the fitted prediction. The prediction term is that rule minus the pooled oracle. $G_d$ is the pooled oracle minus the $V_S^*$ rule, and $G_\sigma$ is the $V_S^*$ rule minus the $V_T^*$ rule, which coincide when $\sigma_1=\sigma_0$. Monte Carlo standard errors are in parentheses, and values computed from Theorem~\ref{thm:pooled-gap} for this design are in brackets.}
\label{tab:decomp}
\small
\setlength{\tabcolsep}{3.5pt}
\begin{tabular}{lrrrrrr}
\toprule
 & \multicolumn{2}{c}{$n\times$ variance} & \multicolumn{4}{c}{Terms of the gap}\\
\cmidrule(lr){2-3}\cmidrule(l){4-7}
Setting & \begin{tabular}[b]{@{}c@{}}Simple\\random\end{tabular} & Adaptive & \begin{tabular}[b]{@{}c@{}}Residual\\variance\end{tabular} & Prediction & $G_d$ & $G_\sigma$\\
\midrule
1 (heteroskedastic) & 29.2 & 21.9 & 1.29 (0.22) & 0.91 (0.16) & 0.37 (0.09) [0.30] & 0\\
2 (heterogeneous effects) & 31.5 & 27.9 & 2.34 (0.32) & 0.74 (0.16) & 0.45 (0.12) [0.50] & 3.19 (0.31) [3.70]\\
3 (misspecified) & 22.5 & 20.1 & 1.60 (0.23) & 3.64 (0.23) & 0.26 (0.07) [0.24] & 0\\
\bottomrule
\end{tabular}
\end{table}

\subsection{Measurements saved}
\label{sec:sim-saving}

To express the gains as numbers of measurements, we repeated the one-update design with expected fractions $\rho$ of 20\%, 30\%, 40\%, and 60\%, using 1000 replications for each value. To reduce Monte Carlo error, each replication recorded the exact variance of $\widehat\Delta$ due to outcome sampling, given the trial data and the fitted rule,
\[
n^{-2}\sum_{i\in\mathcal M}h_{i,n}^2\frac{1-\pi_i}{\pi_i}\{Y_i-\widehat f_1(X_i)\}^2,
\]
instead of drawing $R$. Its mean over replications is the sampling part of the variance. Given the trial data and the fitted prediction, the sampling variance of simple random sampling is proportional to $1/\rho_n-1$. We could therefore compute the fraction that simple random sampling would need to match the adaptive design.

Table~\ref{tab:saving} shows the results. Adaptive sampling saved 20\% to 26\% of the measurements in setting 1, 7\% to 10\% in setting 2, and 8\% to 11\% in setting 3. The total savings were smallest at $\rho=20\%$, where the fixed 10\% initial sample is half of the budget. Among the participants outside the initial sample, the saving was nearly constant for $\rho$ up to 40\%, at about 33\%, 13\%, and 14\% in the three settings, and it was smaller at $\rho=60\%$. For the pooled residual variance, $\kappa_{f_P}=0.84$, 0.55, and 0.68 in the three settings, so an oracle without truncation or initial sample would save 41\%, 23\%, and 32\% (Proposition~\ref{prop:saving}). At $\rho=40\%$ the adaptive design achieved between one third and two thirds of these savings, and the order of the settings matched the order of $\kappa_{f_P}$.

\begin{table}[t]
\centering
\caption{Expected percentage of outcomes that simple random sampling must measure to match the precision of one-update adaptive sampling, with the percentage of measurements saved by adaptive sampling in parentheses. Results are for $n=2000$ and a 10\% initial sample, with 1000 replications per cell. Monte Carlo standard errors are at most 0.12 percentage points for the percentages measured and 0.25 for the savings.}
\label{tab:saving}
\small
\begin{tabular}{lrrrr}
\toprule
 & \multicolumn{4}{c}{Expected percentage measured by adaptive sampling, $100\rho$}\\
\cmidrule(l){2-5}
Setting & 20 & 30 & 40 & 60\\
\midrule
1 (heteroskedastic) & 24.9 (19.6) & 39.8 (24.7) & 54.1 (26.1) & 76.3 (21.4)\\
2 (heterogeneous effects) & 21.4 (6.7) & 33.0 (9.2) & 44.3 (9.7) & 66.4 (9.7)\\
3 (misspecified) & 21.7 (7.8) & 33.4 (10.2) & 45.0 (11.1) & 67.3 (10.9)\\
\bottomrule
\end{tabular}
\end{table}

\subsection{Size of the initial sample}

We also repeated the one-update design with $\rho=40\%$, initial fractions of 2.5\%, 5\%, 10\%, 20\%, and 30\%, and $n=500$, 2000, and 8000, using 1000 replications for each combination and the same exact sampling variance. Figure~\ref{fig:initial} shows a U-shaped curve in every setting. With a very small initial sample the rule is poorly estimated. At $n=500$ with a 2.5\% initial sample (12 outcomes), adaptive sampling was less efficient than simple random sampling in all three settings (Table~\ref{tab:grid}). With a large initial sample, much of the budget is spent before the rule is available. The best fraction fell as $n$ grew. It was 10\% at $n=500$, 5\% or 10\% at $n=2000$, and 5\% at $n=8000$. In setting 1, the excess over the pooled oracle at the best fraction was 8.7, 3.5, and 1.6 for $n=500$, 2000, and 8000. It fell by a factor of 2.5 and then 2.1 for each fourfold increase in $n$, close to the factor of 2 implied by the $n^{-1/2}$ rate in Corollary~\ref{cor:initial-size} when the quadratic bound applies. In setting 3 it fell more slowly (9.2, 5.9, and 4.3), because the misspecified outcome model does not converge to $f_P$ and the prediction term in \eqref{eq:regret2} does not vanish.

\begin{figure}[t]
\centering
\includegraphics[width=\textwidth]{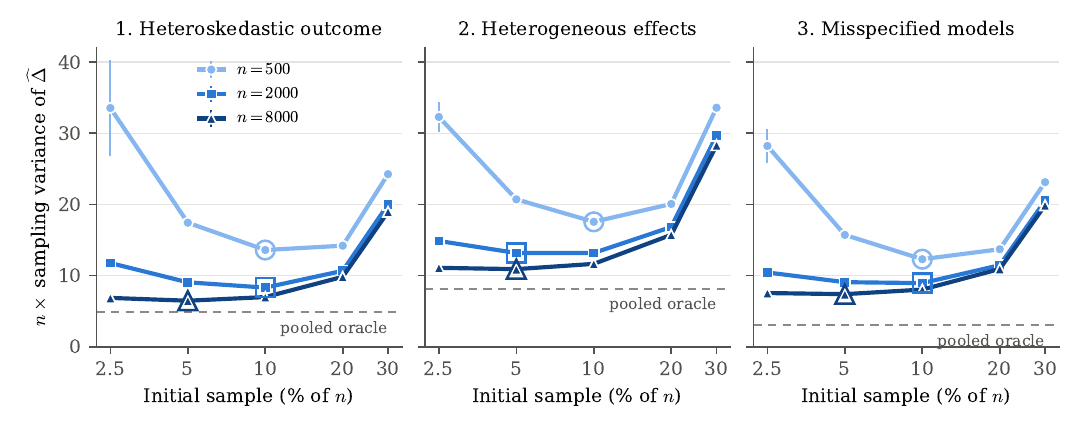}
\caption{Sampling part of the variance of $\widehat\Delta$, multiplied by $n$, for one-update adaptive sampling with 40\% of outcomes measured in expectation, by initial fraction and $n$, with 1000 replications per point. The dashed line is the pooled oracle with the whole budget and no initial sample. Large open symbols mark the best initial fraction for each $n$. Vertical bars are 95\% Monte Carlo intervals and are mostly hidden by the symbols. Values for simple random sampling are in Table~\ref{tab:grid}.}
\label{fig:initial}
\end{figure}

\FloatBarrier
\section{Illustration with a randomized antifungal trial}
\label{sec:toenail}

We used publicly available longitudinal data from a randomized, double-blind comparison of oral terbinafine and itraconazole for toenail onychomycosis \citep{DeBackerEtAl1998}, distributed as the \texttt{toenail} data in the R package HSAUR3 \citep{HothornEveritt2014}. The data contain 1908 visits from 294 participants. We restricted the illustration to the 272 participants with an observed outcome at visit 4, of whom 140 were assigned to terbinafine and 132 to itraconazole. The outcome was moderate or severe onycholysis at visit 4. The auxiliary variables were onycholysis status at visits 1, 2, and 3, indicators that an earlier visit was missing, and the actual times of visits 2 and 3. A missing status was coded 0 together with its indicator, and a missing visit time was replaced by the mean of the observed times at that visit. The auxiliary variables were standardized using the whole cohort, since they are available for everyone. Treatment assignment was not used in the prediction or sampling models.

The aim was to study outcome sampling, and we did not reanalyze the efficacy of the treatments. We treated the visit 4 outcomes of these 272 participants as a fixed complete cohort. The complete-data event rates were 20.7\% for terbinafine and 22.0\% for itraconazole, giving a risk difference of $-0.0126$. In each of 20,000 repetitions, an initial simple random sample of 54 participants (20\%) was selected, and 50\% of the visit 4 outcomes were measured in expectation. The remaining 218 participants were therefore sampled with average probability $82/218=0.376$. A ridge linear probability model was fitted to the initial sample using the auxiliary variables and no treatment information, with the penalty chosen by exact leave-one-out cross-validation within the initial sample, as in Section~\ref{sec:simulation}. Predicted probabilities $\widehat p(X)$ were truncated to $[0.02,0.98]$. For the adaptive design, the remaining sampling probabilities were proportional to $\sqrt{\widehat p(X)\{1-\widehat p(X)\}}$ and restricted to $[0.05,0.95]$. For a binary outcome this is the residual standard deviation implied by the fitted model when the model is well calibrated, so no separate residual variance model was fitted. The simple random design used the same fitted prediction with a constant probability for the remaining participants. As a benchmark that is not feasible in practice, we also fitted the same model to all 272 outcomes and used it both in the estimator and in the sampling rule. This benchmark shows how much is lost by learning the rule from 54 outcomes.

Because the cohort was held fixed, the standard error in this resampling study estimates only the variation due to outcome sampling. For arm $a$, its estimated contribution was
\[
\widehat V_a
=
\frac{1}{n_a^2}
\sum_{i\in\mathcal M:A_i=a}
\frac{R_i\{1-\pi_i\}}{\pi_i^2}
\{Y_i-\widehat f(X_i)\}^2,
\]
which is an unbiased estimate of the sampling variance given the initial sample and the fitted rule.

\begin{table}[t]
\centering
\caption{Resampling results for the antifungal trial with 20,000 repetitions. The target is the complete-data risk difference, $-0.0126$, among the 272 participants with a visit 4 outcome. RE is the relative efficiency against simple random sampling with the same initial sample, with its Monte Carlo standard error in parentheses. The Monte Carlo standard error of coverage is about 0.0015. The full-cohort fit uses all 272 outcomes to build the prediction and the sampling rule.}
\label{tab:toenail}
\small
\begin{tabular}{lrrrrrr}
\toprule
Design & \begin{tabular}[b]{@{}c@{}}Mean\\estimate\end{tabular} & \begin{tabular}[b]{@{}c@{}}Empirical\\SD\end{tabular} & \begin{tabular}[b]{@{}c@{}}Mean\\SE\end{tabular} & Coverage & \begin{tabular}[b]{@{}c@{}}Measured\\fraction\end{tabular} & RE\\
\midrule
Simple random sampling & $-$0.0122 & 0.0405 & 0.0405 & 0.948 & 0.500 & 1.00\\
Adaptive sampling & $-$0.0123 & 0.0315 & 0.0304 & 0.949 & 0.500 & 1.65 (0.02)\\
Full-cohort fit (not feasible) & $-$0.0124 & 0.0235 & 0.0227 & 0.939 & 0.500 & 2.97 (0.04)\\
\bottomrule
\end{tabular}
\end{table}

Both designs reproduced the complete-data risk difference, and the mean estimates were within 1.1 Monte Carlo standard errors of it (Table~\ref{tab:toenail}). Adaptive sampling reduced the empirical standard deviation from 0.0405 to 0.0315. Its relative efficiency was 1.65, a 39\% reduction in sampling variance. The mean estimated standard errors were close to the empirical standard deviations, and coverage was 0.949 for adaptive sampling and 0.948 for simple random sampling. Both designs measured the visit 4 outcome in half of the participants on average. To match the precision of adaptive sampling with 136 measured outcomes, simple random sampling would need 163 (59.8\% of the cohort, Monte Carlo standard error 0.2\%), so adaptive sampling saved about 27 measurements. For the full-cohort fit, $\kappa_f=0.56$, and Proposition~\ref{prop:saving} gives an oracle saving of 24\%, compared with the 16\% achieved. The infeasible benchmark had relative efficiency 2.97, so the adaptive design achieved about 60\% of the variance reduction available to a rule built from the complete data.

\begin{table}[t]
\centering
\caption{Sensitivity of the antifungal trial results to the analysis choices, with 5000 repetitions unless stated otherwise. SD is the empirical standard deviation, SE the mean estimated standard error, and RE the relative efficiency against simple random sampling with the same initial sample, with its Monte Carlo standard error in parentheses. All other choices are as in the primary analysis.}
\label{tab:toenail-sens}
\small
\setlength{\tabcolsep}{4pt}
\begin{tabular}{lrrrrr}
\toprule
 & Simple random & \multicolumn{4}{c}{Adaptive}\\
\cmidrule(lr){2-2}\cmidrule(l){3-6}
Variation & SD & SD & SE & Coverage & RE\\
\midrule
Primary analysis (20,000 repetitions) & 0.0405 & 0.0315 & 0.0304 & 0.949 & 1.65 (0.02)\\
$\widehat p$ truncated to $[0.01,0.99]$ & 0.0412 & 0.0325 & 0.0302 & 0.946 & 1.61 (0.04)\\
$\widehat p$ truncated to $[0.05,0.95]$ & 0.0414 & 0.0334 & 0.0324 & 0.942 & 1.54 (0.03)\\
Penalized logistic regression (2000 repetitions) & 0.0419 & 0.0325 & 0.0315 & 0.940 & 1.66 (0.06)\\
Separate residual variance model & 0.0409 & 0.0348 & 0.0288 & 0.944 & 1.38 (0.05)\\
Initial sample 10\% & 0.0399 & 0.0322 & 0.0304 & 0.948 & 1.54 (0.03)\\
Initial sample 30\% & 0.0466 & 0.0372 & 0.0352 & 0.950 & 1.57 (0.03)\\
\bottomrule
\end{tabular}
\end{table}

Table~\ref{tab:toenail-sens} shows how the results depend on the analysis choices. Changing the truncation of $\widehat p$ or using penalized logistic regression, with its penalty chosen by 5-fold cross-validation, gave relative efficiencies between 1.54 and 1.66. Fitting a separate residual variance model as in Section~\ref{sec:simulation} gave 1.38, and its standard errors were too small (0.0288 against an empirical standard deviation of 0.0348). With 54 binary outcomes the model for log squared residuals is noisy, and the variance $\widehat p(1-\widehat p)$ implied by the outcome model is the better choice. With initial fractions of 10\%, 20\%, and 30\%, the adaptive estimator had standard deviations 0.0322, 0.0315, and 0.0372. The difference between 10\% and 20\% is within about two Monte Carlo standard errors, and 30\% was clearly worse. For this small trial the best fraction appears to be at least as large as in the simulations with $n=500$, as Figure~\ref{fig:initial} suggests.

\begin{figure}[t]
\centering
\includegraphics[width=\textwidth]{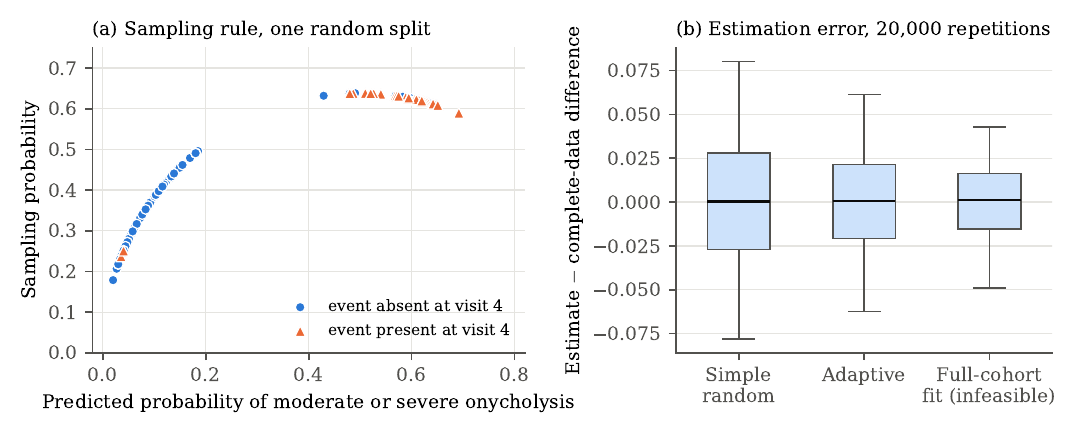}
\caption{Resampling analysis of the antifungal trial. Panel (a) shows the predicted visit 4 event probability and the sampling probability for the 218 participants outside the initial sample in one random split, with triangles for participants who had the event. Panel (b) shows the estimate minus the complete-data risk difference over 20,000 repetitions. Boxes show quartiles, and whiskers show the 2.5th and 97.5th percentiles.}
\label{fig:toenail}
\end{figure}

In the split shown in Figure~\ref{fig:toenail}(a), the area under the receiver operating characteristic curve in the remaining participants was 0.92, and the adaptive sampling probabilities ranged from 0.18 to 0.64. Most participants had predicted probabilities below about 0.07 or near 0.6, because onycholysis status at visit 3 largely determines status at visit 4. The largest sampling probabilities went to predictions near 0.5. The example is retrospective and serves only to study the sampling procedure. It does not address participants whose visit 4 outcome was missing in the original trial, and it does not provide a new clinical comparison of the two treatments.

\FloatBarrier
\section{Discussion}
\label{sec:discussion}
We studied adaptive probability sampling of a costly outcome in a blinded trial, in which the statistician who sets the sampling probabilities does not see treatment assignments. An initial random sample is used to learn where the outcome is hard to predict, later outcomes are sampled with probabilities that depend on the auxiliary variables, and the arms are compared after unblinding with the recorded probabilities. Validity rests on the recorded probabilities, and the working models affect only precision. Estimating the sampling rule costs precision through the prediction error and through the error of the residual variance model. The cost of the residual variance error is linear in that error in general and quadratic without truncation, which affects how large the initial sample should be. Withholding treatment assignment costs further precision. A pooled outcome model treats part of the conditional treatment effect as residual variance, which is a second-order loss near the null and mirrors the inflation of the pooled variance in blinded sample size re-estimation. A common sampling probability also cannot follow different residual variances in the two arms. This loss does not depend on the treatment effect and equals $(\sigma_1-\sigma_0)^2/\rho$ when the residual standard deviations are constant within arms.

In the simulations, adaptive sampling with one update was 12\% to 34\% more efficient than simple random sampling and saved 7\% to 26\% of the measurements, depending on the setting and the fraction measured. The losses relative to the oracle designs followed the terms of Theorems~\ref{thm:adaptive-bound} and \ref{thm:pooled-gap}. Learning the residual variance mattered most when the working models were adequate, prediction error mattered most under misspecification, and the common sampling probability mattered when the residual variances differed between arms. The best initial fraction became smaller as the trial became larger. In the antifungal trial, adaptive sampling reduced the sampling variance by 39\% and saved about one sixth of the measurements, because earlier visits predicted the later binary outcome well. The gains depend on how much the residual standard deviation varies across participants, which Proposition~\ref{prop:saving} makes explicit. The method should not be expected to improve precision in every setting. If the initial sample is very small or the residual variance is hard to predict, a constant probability can do as well or better, in line with the findings of \citet{SfyrakiWang2026}.

The design applies to many kinds of costly outcome, including laboratory assays, central reading of images, adjudication of clinical events, review of medical records, and detailed clinical scores. The auxiliary variables can include baseline measurements, routine follow-up data, or the output of a prediction algorithm. For a binary outcome, participants with predicted event probabilities near one half receive the largest sampling probabilities, so measurements are concentrated where the outcome is least certain. The protocol should list the variables that may enter the sampling rule, the probability bounds, the update times, and the planned fraction measured. The probability and model version applied to each participant should be retained (Algorithm~\ref{alg:design}). Outcomes that must be collected for safety or other clinical reasons can be measured with probability one. Missing outcomes that arise for reasons other than the planned sampling need separate handling and are not covered by the design.

Some limitations still remain. The design fixes the expected number of measured outcomes, while the realized number varies. With independent draws the realized number has standard deviation $\{\sum_i\pi_i(1-\pi_i)\}^{1/2}$, which was at most 20 in our simulations. A fixed number would require a different sampling scheme. Under unequal allocation, the pooled residual variance is not the quantity that minimizes the variance (Remark~\ref{rem:unequal}). The final estimator uses the pooled working prediction. Arm-specific predictions after unblinding may be more efficient but need separate theory (Remark~\ref{rem:arm-specific}). When a separate residual variance model was fitted to few binary outcomes, the standard errors were too small (Section~\ref{sec:toenail}).

The results concern independent participants and a difference in means. Survival outcomes, recurrent events, clustered trials, and multivariate outcomes would need sampling criteria tailored to their estimands. Central review of progression-free survival in a random sample of patients \citep{DoddEtAl2011} is a natural target for an adaptive, blinded version of this design. A further practical question is when to update the sampling rule and how many new outcomes to collect before refitting it. The variance bound is a starting point for that problem, which we do not pursue here.

\appendix
\section{Proofs}
\label{app:proofs}

\begin{proof}[Proof of Theorem~\ref{thm:unbiased}]
Let $\mathcal F_{i-1}=\sigma(\mathcal D_n,S_0,R_1,\ldots,R_{i-1})$. Since $\mathcal H_{i-1}\subseteq\mathcal F_{i-1}$, the functions $f_i$ and $\pi_i$ are $\mathcal F_{i-1}$-measurable, and by \eqref{eq:sampling}
\[
\E\left\{
f_i(X_i)+\frac{R_i}{\pi_i(X_i)}(Y_i-f_i(X_i))\given\mathcal F_{i-1}
\right\}
=f_i(X_i)+Y_i-f_i(X_i)=Y_i.
\]
By iterated expectation, $\E(U_i\mid\mathcal D_n,S_0)=Y_i$. This also holds for $i\in S_0$, where $U_i=Y_i$. Since $n_1$ and $n_0$ are fixed given $\mathcal D_n$, linearity within each arm gives the result. The argument is unchanged if $f_i$ and $\pi_i$ are updated after any number of earlier measurements, provided the current functions are fixed before the current draw.
\end{proof}

\begin{proof}[Proof of Theorem~\ref{thm:clt}]
Subtracting $Y_i$ from \eqref{eq:corrected} gives
\[
U_i-Y_i
=
\left(\frac{R_i}{\pi_i}-1\right)\{Y_i-f_i(X_i)\},
\]
and substitution into the two arm means yields \eqref{eq:decomp-estimator}. Write $T_n=\sqrt n(\Delta_n-\Delta)$, $\xi_i=h_{i,n}(R_i/\pi_i-1)\{Y_i-f_i(X_i)\}$, and $M_n=n^{-1/2}\sum_i\xi_i$. Because $h_{i,n}$ is $\mathcal D_n$-measurable,
\[
\E(\xi_i\mid\mathcal F_{i-1})=0,
\qquad
\E(\xi_i^2\mid\mathcal F_{i-1})=h_{i,n}^2\frac{1-\pi_i}{\pi_i}\{Y_i-f_i(X_i)\}^2.
\]
Thus $M_n$ is a martingale with respect to $\mathcal F_{n,i}=\mathcal F_i$, and its predictable quadratic variation is the left-hand side of \eqref{eq:predvar}. Since $\E\{(R_i/\pi_i-1)^4\mid\mathcal F_{i-1}\}\leq2\epsilon^{-3}$, the conditional Lyapunov sum satisfies
\[
n^{-2}\sum_i\E(\xi_i^4\mid\mathcal F_{i-1})\leq2\epsilon^{-3}\max_ih_{i,n}^4\;n^{-2}\sum_i\{Y_i-f_i(X_i)\}^4=O_p(n^{-1}),
\]
because $\max_ih_{i,n}^4=\max_a(n/n_a)^4=O_p(1)$ and the fourth-moment condition gives $n^{-1}\sum_i\{Y_i-f_i(X_i)\}^4=O_p(1)$. This implies the conditional Lindeberg condition.

Let $\varphi_n(t)=\E\{\exp(itM_n)\mid\mathcal D_n,S_0\}$. The conditional variance in \eqref{eq:predvar} and the conditional Lindeberg sum converge in probability. Every subsequence therefore has a further subsequence along which, for almost every realization of $(\mathcal D_n,S_0)$, they converge in conditional probability. Along it, the martingale central limit theorem \citep[Chapter~3]{HallHeyde1980}, applied under the conditional law and with a constant limit $V_R$ so that no nesting of $\sigma$-fields is needed, gives $\varphi_n(t)\to\exp(-t^2V_R/2)$. Hence $\varphi_n(t)\to\exp(-t^2V_R/2)$ in probability. Because $T_n$ is a function of $\mathcal D_n$,
\[
\E\{\exp(isT_n+itM_n)\}
=\E\{\exp(isT_n)\varphi_n(t)\}
\to\exp(-s^2V_Y/2)\exp(-t^2V_R/2)
\]
by bounded convergence. Thus $(T_n,M_n)$ converges to a pair of independent normal variables, and $\sqrt n(\widehat\Delta-\Delta)=T_n+M_n\overset{d}{\to}N(0,V_Y+V_R)$.

For the standard error, write $e_i=U_i-Y_i$ and $c_i=(1-\pi_i)\pi_i^{-1}\{Y_i-f_i(X_i)\}^2=\E(e_i^2\mid\mathcal F_{i-1})$. Within arm $a$,
\[
s_a^2=S_{Y,a}^2+\frac{1}{n_a}\sum_{A_i=a}c_i+D_a+O_p(n^{-1}),
\qquad
D_a=\frac{1}{n_a}\sum_{A_i=a}(e_i^2-c_i)+\frac{2}{n_a}\sum_{A_i=a}(Y_i-\bar Y_a)e_i-\bar e_a^{\,2},
\]
where $S_{Y,a}^2$ is the sample variance of $Y$ in arm $a$ and $\bar e_a$ is the mean of $e_i$ in arm $a$. The two sums in $D_a$ and $\bar e_a$ are averages of martingale differences with second moments bounded under the fourth-moment condition, so $D_a=o_p(1)$. Since $nS_{Y,a}^2/n_a\to\Var(Y\mid A=a)/p_a$ in probability and $h_{i,n}^2=(n/n_a)^2$ in arm $a$,
\[
n\,\widehat{\operatorname{se}}^2=\sum_a\frac{n}{n_a}s_a^2
=n\sum_a\frac{S_{Y,a}^2}{n_a}+\frac1n\sum_{i=1}^nh_{i,n}^2c_i+o_p(1)
\overset{p}{\longrightarrow}V_Y+V_R
\]
by \eqref{eq:predvar}.
\end{proof}

\begin{lemma}[Variance identity]
\label{lem:variance}
For fixed $f(X)$ and common $\pi(X)$,
\[
\Var(U\mid A=a,X)
=
\sigma_a^2(X)
+
\frac{1-\pi(X)}{\pi(X)}
\E\{(Y-f(X))^2\mid A=a,X\}.
\]
\end{lemma}

\begin{proof}
Conditional on $(A,X,Y)$, $\E(U\mid A,X,Y)=Y$ and
\[
\Var(U\mid A,X,Y)
=
\frac{1-\pi(X)}{\pi(X)}\{Y-f(X)\}^2.
\]
The result follows from the law of total variance.
\end{proof}

\begin{proof}[Derivation of \eqref{eq:variance} and Remark~\ref{rem:unequal}]
For general allocation, the asymptotic variance of the arm difference is $\Var(U\mid A=1)/p_1+\Var(U\mid A=0)/p_0$. By Lemma~\ref{lem:variance} and the law of total variance, $\Var(U\mid A=a)=\Var(Y\mid A=a)+\E[\{1-\pi(X)\}\pi(X)^{-1}\{Y-f(X)\}^2\mid A=a]$. The sampling part is therefore
\[
\sum_{a=0}^1\frac{1}{p_a}\E\left[\frac{1-\pi(X)}{\pi(X)}\{Y-f(X)\}^2\given A=a\right]
=\E\left[\frac{1-\pi(X)}{\pi(X)}\,\frac{\{Y-f(X)\}^2}{p_A^2}\right],
\]
which gives Remark~\ref{rem:unequal}. With $p_1=p_0=1/2$ the right-hand side is $4\E[\{1-\pi(X)\}\pi(X)^{-1}\{Y-f(X)\}^2]=4L_{r_f}(\pi)$, whether $X$ is measured at baseline or after randomization. This proves \eqref{eq:variance}.
\end{proof}

\begin{proof}[Proof of Proposition~\ref{prop:optimal}]
For fixed $r$, minimizing $L_r(\pi)$ is equivalent to minimizing $\E\{r(X)/\pi(X)\}$. Without truncation, the Cauchy--Schwarz inequality gives
\[
\E\left\{\frac{r(X)}{\pi(X)}\right\}\E\{\pi(X)\}
\geq
\left[\E\{\sqrt{r(X)}\}\right]^2,
\]
with equality when $\pi(X)$ is proportional to $\sqrt{r(X)}$, and the budget constraint determines the constant. With the lower and upper bounds, the Karush--Kuhn--Tucker conditions give the clipped solution \eqref{eq:optimal}. The map $c\mapsto\E\min[u,\max\{\epsilon,c\sqrt{r(X)}\}]$ is continuous and nondecreasing. It tends to $\epsilon$ as $c\to0$ and to $u$ as $c\to\infty$ when $\Prb\{r(X)>0\}=1$, so a suitable $c$ exists.
\end{proof}

\begin{proof}[Proof of Proposition~\ref{prop:saving}]
Write $K=K_f$ and $\kappa=\kappa_f$. Without truncation, Proposition~\ref{prop:optimal} gives $\min_{\pi\in\mathcal P_\rho}L_{r_f}(\pi)=K^2/\rho-\E r_f(X)$, and simple random sampling of a fraction $\rho_s$ gives $L_{r_f}(\rho_s)=\E r_f(X)(1/\rho_s-1)$. By \eqref{eq:variance} the two designs have the same variance if and only if $\E r_f(X)/\rho_s=K^2/\rho$. Since $\E r_f(X)=K^2+\Var\{\sqrt{r_f(X)}\}=K^2(1+\kappa^2)$, this is $\rho_s=\rho(1+\kappa^2)$, which proves (i). For (ii), $V_{\mathrm{full}}+4\{K^2/\rho-\E r_f(X)\}=v$ gives the stated value of $\rho$.
\end{proof}

\begin{proof}[Proof of Theorem~\ref{thm:adaptive-bound}]
(i) Write $\pi^*=\pi^*_{r_k}$. Since $L_r(\pi)=\E(r/\pi)-\E r$,
\[
L_{r_k}(\widehat\pi_k)-L_{r_k}(\pi^*)=\E\{(1/\widehat\pi_k-1/\pi^*)\,r_k\},
\qquad
L_{\widehat r_k}(\widehat\pi_k)-L_{\widehat r_k}(\pi^*)=\E\{(1/\widehat\pi_k-1/\pi^*)\,\widehat r_k\}\leq0,
\]
where the inequality holds because $\widehat\pi_k$ minimizes $L_{\widehat r_k}$ over $\mathcal P_\rho$. Subtracting the second identity from the first gives the middle expression in \eqref{eq:regret1}. The lower bound holds because $\pi^*$ minimizes $L_{r_k}$. Both probabilities lie in $[\epsilon,u]$, so $1/\widehat\pi_k$ and $1/\pi^*$ lie in $[1/u,1/\epsilon]$, and $|1/\widehat\pi_k-1/\pi^*|\leq1/\epsilon-1/u\leq C_\epsilon$. This gives the last inequality.

(ii) Write $g=\sqrt{\widehat r_k}$ and $h=\sqrt{r_k}$. Without truncation $\widehat\pi_k=\rho g/\E g$ and $\pi^*=\rho h/\E h$, so
\[
L_{r_k}(\widehat\pi_k)-L_{r_k}(\pi^*)
=\frac{1}{\rho}\left[\E g\;\E\left(\frac{h^2}{g}\right)-(\E h)^2\right].
\]
Put $\delta=h-g$. Then $\E(h^2/g)=\E g+2\E\delta+\E(\delta^2/g)$ and $(\E h)^2=(\E g)^2+2\E g\,\E\delta+(\E\delta)^2$, so the bracket equals $\E g\,\E(\delta^2/g)-(\E\delta)^2\leq\E g\,\E(\delta^2/g)$. This is \eqref{eq:regret-quad}, and the final statement uses $g\geq\sqrt{r_{\min}}$.

(iii) Because $f_P(X)=\E(Y\mid X)$ and the conditional law of $Y$ given $X$ among stage-$k$ participants does not depend on $\mathcal G_k$,
\begin{equation}
r_k(x)
=
r_P(x)+\{\widehat f_k(x)-f_P(x)\}^2.
\label{eq:pythagorean}
\end{equation}
Therefore $L_{r_k}(\pi)\geq L_{r_P}(\pi)$ for every admissible $\pi$, which gives the lower bound in \eqref{eq:regret2}, and
\begin{align*}
L_{r_k}(\pi_{r_k}^*)-L_{r_P}(\pi_{r_P}^*)
&\leq
L_{r_k}(\pi_{r_P}^*)-L_{r_P}(\pi_{r_P}^*)\\
&=\E\left[
\frac{1-\pi_{r_P}^*(X)}{\pi_{r_P}^*(X)}
\{\widehat f_k(X)-f_P(X)\}^2
\right]\\
&\leq C_\epsilon\|\widehat f_k-f_P\|_2^2.
\end{align*}
Adding this bound to the bound on $L_{r_k}(\widehat\pi_k)-L_{r_k}(\pi^*_{r_k})$ from (i) or (ii) gives \eqref{eq:regret2}. Multiplication by four follows from \eqref{eq:variance}.
\end{proof}

\begin{proof}[Proof of Corollary~\ref{cor:repeated}]
Apply \eqref{eq:regret2} conditionally at each stage, multiply by $4w_{k,n}$, and sum over $k$. An initial random sample with fraction $o(1)$ contributes $o_p(1)$ to the average first-order variance. Replacing population normalization by empirical normalization within a stage also contributes $o_p(1)$ under positivity and a law of large numbers for the auxiliary variables. This gives \eqref{eq:repeated-bound}. If its right-hand side converges to zero, the predictable quadratic variation in Theorem~\ref{thm:clt} converges to the pooled oracle value, and Slutsky's theorem gives the oracle first-order variance. For different stage fractions $\rho_k$, let $v(b)=\min_{\pi\in\mathcal P_b}L_{r_P}(\pi)$. If $\pi_1\in\mathcal P_{b_1}$ and $\pi_2\in\mathcal P_{b_2}$, then $\lambda\pi_1+(1-\lambda)\pi_2\in\mathcal P_{\lambda b_1+(1-\lambda)b_2}$. Since $L_{r_P}$ is convex in $\pi$, $v$ is convex and $\sum_kw_{k,n}v(\rho_k)\geq v(\rho)$.
\end{proof}

\begin{proof}[Proof of Corollary~\ref{cor:initial-size}]
Let $\alpha_n=m_n/n$ and $\rho_n=(\rho-\alpha_n)/(1-\alpha_n)$. Conditional on the initial random sample, the initial participants add no sampling variance because their outcomes are measured with certainty. The adaptive stage contributes
\[
4(1-\alpha_n)L_{r_1}(\widehat\pi_1;\rho_n),
\]
where the second argument indicates the sampling fraction. By Theorem~\ref{thm:adaptive-bound} and \eqref{eq:learning-rate}, this quantity differs from the corresponding pooled oracle at fraction $\rho_n$ by $O_p(m_n^{-\gamma})$.

It remains to compare the oracle that reserves a fraction $\alpha_n$ for the initial sample with the oracle that allocates the whole fraction $\rho$ at once. Without truncation, with $K=\E\sqrt{r_P(X)}$,
\[
\min_{\pi:\E\pi=b} L_{r_P}(\pi)
=
\frac{K^2}{b}-\E r_P(X).
\]
A direct calculation gives
\[
(1-\alpha_n)
\min_{\pi:\E\pi=\rho_n}L_{r_P}(\pi)
-
\min_{\pi:\E\pi=\rho}L_{r_P}(\pi)
=\alpha_n\left\{\E r_P(X)-K^2\frac{2\rho-1-\rho\alpha_n}{\rho(\rho-\alpha_n)}\right\}.
\]
Writing $\E r_P(X)=K^2+\Var\{\sqrt{r_P(X)}\}$, the right-hand side is $\alpha_n[K^2(1-\rho)^2/\rho^2+\Var\{\sqrt{r_P(X)}\}]+O(\alpha_n^2)$. Combining the two parts gives
\[
O_p(\alpha_n+m_n^{-\gamma})
=
O_p\left(\frac{m_n}{n}+m_n^{-\gamma}\right).
\]
Balancing $m_n/n$ and $m_n^{-\gamma}$ gives $m_n^{1+\gamma}\asymp n$, which proves the stated rate. The examples follow from \eqref{eq:regret1} and \eqref{eq:regret-quad}, using $|\sqrt a-\sqrt b|\leq|a-b|/\sqrt b$ for $a\geq0$ and $b>0$.
\end{proof}

\begin{proof}[Proof of \eqref{eq:decomposition}]
Condition on $X$ and apply the law of total variance with respect to treatment,
\begin{align*}
\Var(Y\mid X)
&=
\E\{\Var(Y\mid A,X)\mid X\}
+
\Var\{\E(Y\mid A,X)\mid X\}\\
&=
\eta\sigma_1^2+(1-\eta)\sigma_0^2
+
\eta(1-\eta)(q_1-q_0)^2.
\end{align*}
\end{proof}

\begin{proof}[Proof of Theorem~\ref{thm:pooled-gap}]
First allow the sampling probability to depend on treatment and use the arm-specific prediction $q_a$. The part of the variance that depends on sampling is
\[
2\sum_{a=0}^1
\E\left\{
\frac{1-\pi_a(X)}{\pi_a(X)}\sigma_a^2(X)
\right\},
\]
subject to $\E\{\pi_1(X)+\pi_0(X)\}/2=\rho$. Without truncation the minimizer is $\pi_a(X)=c\sigma_a(X)$ with
\[
c=\frac{2\rho}{\E\sigma_1(X)+\E\sigma_0(X)},
\]
and its sampling contribution is
\[
\frac{\{\E\sigma_1(X)+\E\sigma_0(X)\}^2}{\rho}
-2\E\{\sigma_1^2(X)+\sigma_0^2(X)\}.
\]

If one probability $\pi(X)$ is required while the arm-specific predictions are kept, the sampling contribution is $4L_s(\pi)$, with minimum
\[
\frac{4\{\E\sqrt{s(X)}\}^2}{\rho}
-4\E\{s(X)\}.
\]
Since $4\E\{s(X)\}=2\E\{\sigma_1^2(X)+\sigma_0^2(X)\}$, subtraction gives \eqref{eq:gsigma}. The pointwise inequality $\sigma_1(x)+\sigma_0(x)\leq2\sqrt{s(x)}$ implies $G_\sigma\geq0$, with equality if and only if $\sigma_1=\sigma_0$ almost surely. If $\sigma_1$ and $\sigma_0$ are constants, $4s-(\sigma_1+\sigma_0)^2=2\sigma_1^2+2\sigma_0^2-(\sigma_1+\sigma_0)^2=(\sigma_1-\sigma_0)^2$.

When treatment is also omitted from the outcome prediction, $r_P=s+d^2/4$ by \eqref{eq:decomposition}. For every admissible $\pi$, $L_{r_P}(\pi)\geq L_s(\pi)$, which proves $G_d\geq0$ and \eqref{eq:gap-decomp}. Without truncation, $\min_{\pi}L_r(\pi)=\{\E\sqrt{r(X)}\}^2/\rho-\E r(X)$ and $\E r_P(X)-\E s(X)=\E\{d^2(X)\}/4$, which gives the closed form in \eqref{eq:gd}. Evaluating the $r_P$ objective at $\pi_s^*$ gives
\[
G_d\leq4\{L_{r_P}(\pi_s^*)-L_s(\pi_s^*)\}
=\E\left[\frac{1-\pi_s^*(X)}{\pi_s^*(X)}\,d^2(X)\right]\leq C_\epsilon\|d\|_2^2,
\]
which is \eqref{eq:gd-bound}. For the first-order statement, let $a=\E\sqrt s$ and $b=\E(\sqrt{r_P}-\sqrt s)$. Since $\sqrt{r_P}-\sqrt s=(d^2/4)/(\sqrt{r_P}+\sqrt s)=d^2/(8\sqrt s)+O(d^4/s_{\min}^{3/2})$,
\[
G_d=\frac4\rho(2ab+b^2)-\E(d^2)
=\frac a\rho\E\left(\frac{d^2}{\sqrt s}\right)-\E(d^2)+O(\|d\|_\infty^4)
=\E\left[\frac{1-\pi_s^*(X)}{\pi_s^*(X)}d^2(X)\right]+O(\|d\|_\infty^4),
\]
because $1/\pi_s^*=a/(\rho\sqrt s)$. Under local alternatives with $\|d_n\|_2=O(n^{-1/2})$, \eqref{eq:gd-bound} gives $G_d=O(n^{-1})$.
\end{proof}

\section{Additional simulation results}

\begin{table}[H]
\centering
\caption{Sampling part of the variance of $\widehat\Delta$, multiplied by $n$, for one-update adaptive sampling and simple random sampling (adaptive / simple random), by initial fraction and $n$, with 1000 replications per cell and 40\% of outcomes measured in expectation. These are the values plotted in Figure~\ref{fig:initial}.}
\label{tab:grid}
\small
\begin{tabular}{lrccccc}
\toprule
 & & \multicolumn{5}{c}{Initial sample (\% of $n$)}\\
\cmidrule(l){3-7}
Setting & $n$ & 2.5 & 5 & 10 & 20 & 30\\
\midrule
1 (heteroskedastic) & 500 & 33.5 / 25.7 & 17.4 / 18.8 & 13.6 / 18.0 & 14.2 / 21.9 & 24.2 / 37.7\\
 & 2000 & 11.8 / 15.8 & 9.1 / 15.2 & 8.3 / 15.8 & 10.7 / 20.5 & 20.0 / 35.4\\
 & 8000 & 6.9 / 13.7 & 6.5 / 14.0 & 7.0 / 15.2 & 9.8 / 20.0 & 19.0 / 35.0\\
\addlinespace
2 (heterogeneous effects) & 500 & 32.3 / 24.6 & 20.7 / 19.8 & 17.6 / 18.9 & 20.1 / 23.1 & 33.6 / 39.0\\
 & 2000 & 14.9 / 16.1 & 13.2 / 15.5 & 13.2 / 16.5 & 16.8 / 21.1 & 29.7 / 36.7\\
 & 8000 & 11.1 / 14.1 & 10.9 / 14.4 & 11.7 / 15.7 & 15.7 / 20.6 & 28.3 / 36.0\\
\addlinespace
3 (misspecified) & 500 & 28.2 / 20.1 & 15.7 / 15.0 & 12.3 / 13.5 & 13.7 / 16.1 & 23.1 / 27.2\\
 & 2000 & 10.4 / 11.6 & 9.1 / 11.0 & 8.9 / 11.5 & 11.5 / 14.8 & 20.5 / 25.7\\
 & 8000 & 7.6 / 9.9 & 7.4 / 10.0 & 8.0 / 10.9 & 10.9 / 14.5 & 19.9 / 25.3\\
\bottomrule
\end{tabular}
\end{table}


\bibliographystyle{plainnat}
\bibliography{ref_final}

\end{document}